\documentclass[submission,copyright,creativecommons]{eptcs}
\providecommand{\event}{GandALF 2022} % Name of the event you are submitting to

\usepackage{iftex}

\ifpdf
  \usepackage{underscore}         % Only needed if you use pdflatex.
  \usepackage[T1]{fontenc}        % Recommended with pdflatex
\else
  \usepackage{breakurl}           % Not needed if you use pdflatex only.
\fi

\usepackage{amssymb}
\usepackage{amsthm}
\usepackage{stmaryrd}
\usepackage{mathtools}
\usepackage{xspace}
\usepackage{mathpartir}
\usepackage{tikz}
\usepackage{hyperref}

\title{Capturing Bisimulation-Invariant Exponential-Time Complexity Classes}
\author{Florian Bruse
\institute{University of Kassel\\ Kassel, Germany}
\email{florian.bruse@uni-kassel.de}
\and
David Kronenberger 
\institute{University of Kassel \\ Kassel, Germany}
\and Martin Lange
\institute{University of Kassel\\ Kassel, Germany}
\email{martin.lange@uni-kassel.de}
}
\def\titlerunning{Capturing Bisimulation-Invariant Exponential-Time Complexity Classes}
\def\authorrunning{F. Bruse, D. Kronenberger \& M. Lange}
\begin{document}
\maketitle

% !TEX root = main.tex

\theoremstyle{plain}
\newtheorem{theorem}{Theorem}
\newtheorem{lemma}[theorem]{Lemma}
\newtheorem{proposition}[theorem]{Proposition}

\theoremstyle{definition}
\newtheorem{definition}[theorem]{Definition}
\newtheorem{example}[theorem]{Example}

\theoremstyle{remark}
\newtheorem{remark}[theorem]{Remark}
\newtheorem{observation}[theorem]{Observation}

\renewcommand{\epsilon}{\varepsilon}

\newcommand{\Nat}{\ensuremath{\mathbb{N}}}
\newcommand{\PHFL}{\text{PHFL}\xspace}
\newcommand{\HFL}{\text{HFL}\xspace}
\newcommand{\phfl}[2]{\PHFL^{#1}_{#2}}
\newcommand{\hfl}[1]{\HFL^{#1}}
\newcommand{\HOLFP}{\holfp{}}
\newcommand{\holfp}[1]{\ensuremath{\textsf{HO}^{#1}\textsf{(LFP)}}\xspace}
\newcommand{\mucalc}{\ensuremath{\mathcal{L}_\mu}\xspace}
\newcommand{\polymucalc}{\ensuremath{\mucalc^\omega}\xspace}

\newcommand{\bisim}[1]{#1\ensuremath{\!\!{/}\!_{\sim}}}

\newcommand{\ExpTime}{\textsf{EXPTIME}\xspace}
\newcommand{\ExpSpace}{\textsf{EXPSPACE}\xspace}
\newcommand{\PTime}{\textsf{P}\xspace}
\newcommand{\PSpace}{\textsf{PSPACE}\xspace}
\newcommand{\NPTime}{\textsf{NP}\xspace}
\newcommand{\NLogSpace}{\textsf{NLOGSPACE}\xspace}

\newcommand{\Transsys}{\ensuremath{{T}}}
\newcommand{\States}{\ensuremath{{S}}}
\newcommand{\Prop}{\ensuremath{\mathbf{P}}}
\newcommand{\Actions}{\ensuremath{\mathbf{A}}}
\newcommand{\Label}{\ensuremath{\ell}}
\newcommand{\state}{\ensuremath{s}}
\newcommand{\statealt}{\ensuremath{t}}

\def\newarrow#1{\mathop{{\hbox{\setbox0=\hbox{$\scriptstyle{#1\quad}$}{$%
\mathrel{\mathop{\setbox1=\hbox to
\wd0{\rightarrowfill}\ht1=3pt\dp1=-2pt\box1}\limits^{#1}}%
$}}}}}

\def\newarrowi#1#2{\mathop{{\hbox{\setbox0=\hbox{$\scriptstyle{#1\quad}$}{$%
\mathrel{\mathop{\setbox1=\hbox to
\wd0{\rightarrowfill}\ht1=3pt\dp1=-2pt\box1}\limits^{#1}}_{#2}%
$}}}}}

\newcommand{\Transition}[3]{\ensuremath{#1 \newarrow{#2} #3}}
\newcommand{\Transitioni}[4]{\ensuremath{#1 \newarrowi{#2}{#4} #3}}

\newcommand{\indtype}{\odot}
\newcommand{\ORDHO}{\mathit{order}}
\newcommand{\ordho}[1]{\ORDHFL(#1)}
\newcommand{\hovars}{\ensuremath{\mathbf{V}}}
\newcommand{\lfp}{\ensuremath{\mathit{lfp}}}
\newcommand{\gfp}{\ensuremath{\mathit{gfp}}}
\newcommand{\LFP}{\ensuremath{\mathit{LFP}}}

\newcommand{\grtype}{\bullet}
\newcommand{\varpos}{\ensuremath{+}}
\newcommand{\varneg}{\ensuremath{-}}
\newcommand{\varboth}{\ensuremath{0}}

\newcommand{\dimension}{\ensuremath{d}}
\newcommand{\ORDHFL}{\mathit{ord}}
\newcommand{\ordhfl}[1]{\ORDHFL(#1)}

\newcommand{\sem}[3]{\ensuremath{\llbracket #1 \rrbracket^{#2}_{#3}}}
\newcommand{\semm}[2]{\sem{#1}{#2}{}}

\newcommand{\mytrue}{\ensuremath{\mathtt{t\!t}}}
\newcommand{\myfalse}{\ensuremath{\mathtt{f\!f}}}
\newcommand{\lvars}{\ensuremath{\mathbf{L}}}
\newcommand{\fvars}{\ensuremath{\mathbf{F}}}
\newcommand{\mudiam}[2]{\langle #1 \rangle #2}
\newcommand{\mubox}[2]{[#1] #2}
\newcommand{\subst}[2]{\{#1\}#2}

\newcommand{\trule}[2]{\inferrule{#1}{#2}}
\newcommand{\judg}[2]{#1\vdash #2}

\newcommand{\nxt}{\mathit{next}}

\newcommand{\SEMTRANS}{\mathit{tptr}}
\newcommand{\semtrans}[3]{\SEMTRANS^{#2}_{#3}(#1)}
\newcommand{\SYNTRANS}{\mathit{trans}}
\newcommand{\syntrans}[1]{\SYNTRANS(#1)}

\newcommand{\TODOcomment}[2]{%
  \stepcounter{TODOcounter#1}%
  {\scriptsize\bf$^{(\arabic{TODOcounter#1})}$}%
  \marginpar[\fbox{
    \parbox{2cm}{\raggedleft
      \scriptsize$^{({\bf{\arabic{TODOcounter#1}{#1}}})}$%
      \scriptsize #2}}]%
  {\fbox{\parbox{2cm}{\raggedright 
      \scriptsize$^{({\bf{\arabic{TODOcounter#1}{#1}}})}$%
      \scriptsize #2}}}
}%

\newcounter{TODOcounter}
\newcommand{\TODO}[1]{\TODOcomment{}{#1}}

%%% Local Variables:
%%% mode: latex
%%% TeX-master: "main"
%%% End:

\allowdisplaybreaks

\begin{abstract}
Otto's Theorem characterises the bisimulation-invariant PTIME queries over graphs as exactly
those that can be formulated in the polyadic $\mu$-calculus, hinging on the Immerman-Vardi Theorem
which characterises PTIME (over ordered structures) by First-Order Logic with least fixpoints. 
This connection has been extended to characterise bisimulation-invariant EXPTIME by an extension 
of the polyadic $\mu$-calculus with functions on predicates, making use of Immerman's characterisation
of EXPTIME by Second-Order Logic with least fixpoints. 

In this paper we show that the bisimulation-invariant versions of all classes in the exponential 
time hierarchy have logical counterparts which arise as extensions of the polyadic $\mu$-calculus 
by higher-order functions. This makes use of the characterisation of $k$-EXPTIME by Higher-Order
Logic (of order $k+1$) with least fixpoints, due to Freire and Martins.
\end{abstract}

\section{Introduction}
% !TEX root = main.tex

Descriptive complexity theory aims at characterising complexity classes -- usually defined via computational
resources like time or space consumption -- by means of \emph{logical} resources. Its central notion is that
of a complexity class $\mathcal{C}$ being \emph{captured} by a logic $\mathcal{L}$ in the sense that the
properties which can be checked in complexity $\mathcal{C}$ are exactly those that can be defined in $\mathcal{L}$.
This provides a characterisation of computational complexity that is independent of a machine model. Instead,
computational difficulty is characterised by the need for particular logical resources like $k$-order quantifiers
or fixpoints of a particular type. It is widely believed that this provides a more promising line of attack for
notoriously difficult problems of separating complexity classes; at least it makes machinery that is traditionally
used for measuring the expressive power of logics available for such tasks.

Here we adopt terminology from database theory as this is traditionally close to descriptive complexity, and speak 
of \emph{queries} being \emph{answered} instead of problems being solved or languages being decided, different names 
for the same thing.

Ever since Fagin's seminal work showing that the complexity class \NPTime is captured by Existential Second-Order Logic
\cite{Fagin74}, descriptive complexity has provided logical characterisations of many standard complexity classes. 
The Abiteboul-Vianu Theorem for instance states that \PSpace is captured by First-Order Logic with Partial Fixpoint
Operators \cite{Abiteboul87a}. The proofs of these results rely on the existence of a total order on the structure
at hand. This is not a restriction for characterisations of complexity classes including and above NP as the resources
available there are sufficient to construct such an order. 

There is no known way to define or construct such a total order in deterministic polynomial time which is a major 
obstacle for capturing the complexity class \PTime. There is some belief that such a logic should exist, possibly 
in the form of first-order logic with additional operators like fixpoints, counting and others, cf.\ \cite{DBLP:conf/lics/Grohe08}.
This is grounded in the characterisation of the complexity class of \emph{Ordered Polynomial Time} -- i.e.\ those 
queries that can be answered in polynomial time on structures that are equipped with a total order -- by First-Order 
Logic with Least Fixpoints \cite{DBLP:conf/stoc/Vardi82,DBLP:journals/iandc/Immerman86}.

An interesting result was then found by Otto who considered another restriction of the class \PTime, namely that
of \emph{bisimulation-invariant} queries (on graphs, naturally). He showed that this class, denoted \bisim{\PTime}
is captured by the polyadic $\mu$-calculus \polymucalc\ \cite{Otto99}, a generalisation of the well-known modal 
$\mu$-calculus to interpretations of formulas not in states but in tuples of states of fixed arity \cite{AndersenPMC:1994}. 
This is particularly interesting as it shifts the borderline at which the availability of an order becomes critical, 
from between \PTime and \NPTime to between \bisim{\PTime} and \bisim{\NLogSpace}. 

This opens up the question of further capturing results of bisimulation-invariant complexity classes by (modal)
logics. Indeed, characterisations have been found for \bisim{\ExpTime} and \bisim{\PSpace} in terms an extension of
\polymucalc by first-order functions from predicates to predicates \cite{conf/ifipTCS/LangeL14}. The logic capturing 
\bisim{\ExpTime} is coined $\phfl{1}{}$ -- Polyadic Higher-Order Fixpoint Logic of order 1. The higher-order extension is borrowed from \HFL 
which extends the modal $\mu$-calculus with a simply typed $\lambda$-calculus \cite{DBLP:conf/concur/ViswanathanV04}.
A syntactical restriction called \emph{tail-recursiveness} \cite{conf/rp/BruseLL17} has been identified that captures 
\bisim{\PSpace} \cite{conf/ifipTCS/LangeL14}, and when applying this restriction to \polymucalc or, likewise, $\phfl{0}{}$, 
one captures \bisim{\NLogSpace} (with the help of a particular partial order only, though) \cite{conf/ifipTCS/LangeL14}.

Considering bisimulation-invariant complexity classes above \bisim{\NPTime} does not serve the same purpose as it does
for smaller ones as one of the key motivations for moving to the bisimulation-invariant world is to avoid the need for
a total order. On the other hand, for two complexity classes $\mathcal{C}, \mathcal{C}'$ that have complete and
bisimulation-invariant problems we have $\mathcal{C} \ne \mathcal{C}'$ iff $\bisim{\mathcal{C}} \ne \bisim{\mathcal{C}'}$.
Most of the standard complexity classes posses such problems, for example (1-letter) NFA universality for \NPTime, resp.\
\PSpace, unbounded tree automaton intersection for \ExpTime, etc. Moreover, separating bisimulation-invariant classes
may be easier due to their close connection to modal logics where separating their expressiveness is routinely done (not
for the relatively complex higher-order modal fixpoint logics mentioned here, though).

In this paper we extend the characterisation of bisimulation-invariant time complexity classes to all the levels of the
exponential time hierarchy. We show that \bisim{$k$-\ExpTime} is captured by $\phfl{k}{}$, the polyadic version of the
higher-order extension of the modal $\mu$-calculus with functions up to type order $k$. We remark that a similar 
characterisation of the space complexity classes $k$-\ExpSpace is also possible \cite{KronenbergerMSc19} but needs to 
be omitted for lack of space and is therefore left for a future publication.

The paper is organised as follows. In Sect.~\ref{sec:prel} we recall the necessary preliminaries, mainly about the logics
studied here. In Sect.~\ref{sec:upper} we provide the easy half of the capturing result by putting together known results
and constructions which witness that any $\phfl{k}{}$ query can be answered in $k$-fold exponential time. In 
Sect.~\ref{sec:quantifier} we prepare for the more difficult and other half of the capturing result. We rely on a logical
characterisation of $k$-\ExpTime in terms of Higher-Order Predicate Logic with Fixpoints \cite{FREIRE201171}, and a 
key step in expressing such (bisimulation-invariant) queries in Higher-Order Modal Fixpoint is the modelling of 
higher-order quantification using an enumeration technique. In Sect.~\ref{sec:lower} we put this to use for showing
that any \bisim{$k$-\ExpTime}-query is definable in $\phfl{k}{}$. In Sect.~\ref{sec:concl} we conclude with remarks on
further work etc.

\section{Preliminaries}
\label{sec:prel}
% !TEX root =  main.tex

Let $\Prop$ and $\Actions$ be finite sets of \emph{propositions}, resp.\ \emph{actions}. 
A labelled transition system (LTS) is a tuple $\Transsys = (\States, \{\Transition{}{a}{}\}_{a\in\Actions},\Label)$
where $\States$ is a nonempty set of \emph{states},  $\Transition{}{a}{} \subseteq \States  \times \States$ 
is a \emph{transition relation} for each action $a \in \Actions$, and $\Label\colon\States \to 2^{\Prop}$ labels the states
with those propositions that are true in them. 
We write $\Transition{\state}{a}{\statealt}$ instead of $(\state, \statealt) \in \Transition{}{a}{}$.

For $d \ge 1$, a $d$-\emph{pointed LTS} is a  pair $\Transsys, (\state_1,\dotsc,\state_\dimension)$ of an LTS and a $d$-tuple 
of states in it. It is also simply called \emph{pointed} when $d$ is clear from the context. Note that in the case of $d=1$,
this coincides with the usual notion of a pointed LTS. Given some tuple $\overline{\state} = (\state_1,\dotsc,\state_\dimension)$
and $i \leq \dimension$, we write $\overline{\state}[\statealt/i]$ to denote the tuple 
$(\state_1,\dotsc,\state_{i-1},\statealt,\state_{i+1},\dotsc,\state_\dimension)$.

\subsection{Polyadic Higher-Order Fixpoint Logic}

We assume familiarity with the modal $\mu$-calculus \mucalc \cite{Kozen83}. Polyadic Higher-Order Fixpoint Logic \PHFL 
\cite{conf/ifipTCS/LangeL14} extends \mucalc in several ways: (\textsc{i}) \HFL \cite{DBLP:conf/concur/ViswanathanV04} adds 
to it a simply typed $\lambda$-calculus; (\textsc{ii}) the polyadic $\mucalc$,  \polymucalc \cite{AndersenPMC:1994,Otto99} 
is obtained by lifting the interpretation of formulas in states to tuples of states of fixed arity. Now \PHFL merges both 
extensions. Its introduction requires a few technicalities.

\paragraph{Types.}

Types are used to govern the syntax of \PHFL, especially those parts that denote functions. They are derived from the grammar
\[
\tau \Coloneqq \grtype \mid \tau^v \to \tau
\]
where $\grtype$ is a ground type for propositions, $v \in \{\varpos, \varneg, \varboth\}$ are \emph{variances} denoting 
whether a function is monotonically increasing, monotonically decreasing, or constant in its argument. Variances are 
only needed to check well-typedness; we often do not display them for the sake of readability. 

The \emph{order} $\ORDHFL$ of a type is defined via $\ordhfl{\grtype} = 0$ and 
$\ordhfl{\tau_1 \to \tau_2} = \max(\ordhfl{\tau_1}+1, \ordhfl{\tau_2})$.
Types associate to the right whence every type is of the form $\tau_1 \to \dotsb \to \tau_n \to \grtype$.

Given some LTS $\Transsys = (\States, \{\Transition{}{a}{}\}_{a\in\Actions},\Label)$ and some $0 < \dimension \in \Nat$,
the semantics $\semm{\tau}{\Transsys}$ of a type $\tau$ is a complete lattice, defined inductively as follows:
\begin{align*}
\semm{\grtype}{\Transsys} &= (2^{\States^{\dimension}}, \subseteq) &
\semm{\tau_2^v \to \tau_1}{\Transsys} &= (\semm{\tau_2}{\Transsys} \to \semm{\tau_1}{\Transsys}, \sqsubseteq_{\tau_2^v \to \tau_1})
\end{align*}
where we tacitly identify a lattice with its domain and where the order $\sqsubseteq_{\tau_2^v \to \tau_1}$ is given by 
\[ f \sqsubseteq_{\tau_2^v \to \tau_1} g \quad \Leftrightarrow \quad \text{for all } x \in \semm{\tau_2}{\Transsys} \text{ we have } 
\left\{ \begin{aligned} 
f(x) \sqsubseteq_{\tau_1} g(x) &, \text{ if } v = \varpos \\
g(x) \sqsubseteq_{\tau_1} f(x) &, \text{ if } v = \varneg \\
f(x) = g(x) &, \text{ if }  v = \varboth.
 \end{aligned}\right.\]
Hence, the semantics of $\grtype$ is the powerset lattice over $\States^{\dimension}$ and $\tau_2^v \to \tau_1$
is the lattice of all monotonically increasing, monotonically decreasing or constant functions 
from $\semm{\tau_2}{\Transsys}$ to $\semm{\tau_1}{\Transsys}$,
depending on $v$. Such a set together with the above point-wise order forms a complete lattice if
$\semm{\tau_1}{\Transsys}$ is one. Hence,
monotone functions in such a lattice always have a least and greatest fixpoint due to the 
Knaster-Tarski-Theorem.
We write $\bigsqcup_{\tau_2^v \to \tau_1}$ to denote the join operator in this lattice.

Note that, technically, $\semm{\tau_1}{\Transsys}$ is also parameterised in $\dimension$. However,
since $\dimension$ is usually clear from context, we do not display it to avoid clutter.

\paragraph{Syntax.}

Let $\Prop$ and $\Actions$ be as above and let $\fvars$ and $\lvars$ be finite sets of \emph{fixpoint variables}, resp.\ 
\emph{lambda variables}. We use upper case letters $X, Y, \dots$ for the former and lower case letters $x, y, f, g$ for the latter.

Let $d \ge 1$. By $[d]$ we denote the set $1,\ldots,d$. The set of -- potentially non-well-formed -- formulas of the $\dimension$-adic fragment of
\PHFL, called $\phfl{}{\dimension}$, is derived via
\[
\varphi \Coloneqq p_i \mid \varphi \vee \varphi \mid \neg \varphi \mid \mudiam{a_i}{\varphi} \mid \subst{\sigma}{\varphi} 
\mid \lambda(x \colon \tau).\ \varphi \mid x \mid (\varphi\, \varphi) \mid \mu (X\colon \tau).\, \varphi \mid X
\]
where $x \in \lvars$, $X \in \fvars$, $p \in \Prop$, $a \in \Actions$, $1 \leq i \leq d$ and 
$\sigma\colon[d] \to [d]$ is a mapping on so-called \emph{indices}. We assume that standard connectives
such as $\myfalse$, $\wedge, \mubox{a}{}$ and $\nu (X\colon\tau)$ are available via the obvious dualities if needed. We will also
allow ourselves to use more convenient notation for the definition of (more complex) functions and their applications.
For instance, $(\ldots((\varphi\,\psi_1)\,\psi_2)\ldots)\, \psi_n$ is simply written as 
$\varphi(\psi_1,\ldots,\psi_n)$, and $\lambda (x_1\colon\tau_1).\lambda (x_2\colon\tau_2).\ldots \lambda (x_n\colon\tau_n).\psi$
is written as $\lambda (x_1\colon \tau_1,\ldots,x_n\colon\tau_n).\psi$, or even $\lambda(x_1,\dotsc,x_n\colon\tau).\psi$ 
if $\tau = \tau_i$ for all $1 \leq i \leq n$.

The notions of free and bound variables in a formula are as usual, with $\lambda (x\colon\tau).\varphi$ binding $x$
and with $\mu (X\colon\tau).\ \varphi$ binding $X$. 

Over an LTS $\Transsys$, a $\phfl{}{d}$ formula intuitively defines a set of $d$-tuples in $\Transsys$, or a function that 
transforms such a set into such a set, or into a function etc., depending on the formula's type which will be explained shortly. 
Before that we briefly introduce the intuitive meaning of the operators in the syntax. The first five listed in the grammar
above all define a subset of $\States^d$. For example, $p_i$ denotes all tuples such that $p$ holds at their $i$th state.
The operator $\subst{\sigma}{\varphi}$ rearranges the positions in tuples in such a subset. The modality $\mudiam{a_i}\varphi$
expresses that $\varphi$ holds on a tuple after replacing the $i$th state in it by some $a$-successor. A formula
of the form $\lambda(x\colon\tau^v).\, \varphi$ defines a function that consumes an argument of type $\tau$ and
is monotonically increasing, monotonically decreasing, or constant in this argument, depending on $v$. A formula of the form $(\varphi\, \psi)$ 
denotes the application of the semantics of $\varphi$ to the object defined by $\psi$. Finally, fixpoints
can now also define higher-order functions.

Obviously, not all formulas that can be derived from the above grammar can be given a semantics in a meaningful way;  
consider e.g.\ $(p_1\,p_2)$, Moreover, as is typical in a situation involving least and greatest fixpoints, 
the use of negation has to be restricted, cf.\ the
example $\mu (X\colon\grtype).\ \neg X$. Hence, \PHFL has a type system to filter out formulas that cannot be endowed with
a proper semantics.

\begin{figure}[t]
\begin{mathpar}

\trule{ }{\judg{\Gamma}{p_i\colon\grtype}}

\trule{\judg{\Gamma}{\varphi_1\colon\grtype}\\\judg{\Gamma}{\varphi_2\colon\grtype}}
{\judg{\Gamma}{\varphi_1 \vee \varphi_2\colon\grtype}}

\trule{\judg{{}\overline\Gamma}{\varphi\colon\grtype}}{\judg{\Gamma}{\neg\varphi\colon:\grtype}}

\trule{\judg{\Gamma}{\varphi:\grtype}}{\judg{\Gamma}{\mudiam{a_i}{\varphi}\colon\grtype}}

\trule{\judg{\Gamma}{\varphi:\grtype}}{\judg{\Gamma}{\subst{\sigma}{\varphi}\colon\grtype}}

\trule{\judg{\Gamma,x^v\colon\tau_1}{\varphi\colon\tau_2}}
{\judg{\Gamma}{\lambda (x^v:\tau_1).\ \varphi \colon\tau_1^v\to\tau_2}}

\trule{v\in\{\varpos,\varboth\}}{\judg{\Gamma\;,\;x^v\colon \tau}{x\colon\tau}}

\trule{\judg{\Gamma,X^{\varpos}\colon\tau}{\varphi\colon\tau}}
{\judg{\Gamma}{\mu (X\colon\tau).\ \varphi\colon\tau}}

\trule{ }{\judg{\Gamma\;,\;X^\varpos\colon\tau}{X\colon\tau}}

\trule{\judg{\Gamma}{\varphi_1:\tau_2^{\varpos}\to\tau_1}\\\judg{\Gamma}{\varphi_2:\tau_2}}
{\judg{\Gamma}{(\varphi_1\ \varphi_2) \colon\tau_1}}

\trule{\judg{\Gamma}{\varphi_1:\tau_2^{\varneg}\to\tau_1}\\\judg{\overline{\Gamma}}{\varphi_2:\tau_2}}
{\judg{\Gamma}{(\varphi_1\ \varphi_2) \colon\tau_1}}

\trule{\judg{\Gamma}{\varphi_1:\tau_2^{\varboth}\to\tau_1}\\
\judg{\Gamma}{\varphi_2\colon\tau_2}\\\judg{\overline\Gamma}{\varphi_2\colon\tau_2}
}
{\judg{\Gamma}{\varphi_1\ \varphi_2 \colon\tau_1}}
\end{mathpar}
\vspace*{-6mm}
\caption{The \PHFL~typing system.} 
\label{fig:typing-rules}
\end{figure}

A finite sequence $\Gamma$ of \emph{hypotheses} of the form $X^{v}\colon\tau$ or $x^v\colon\tau$ in which each variable occurs 
at most once is called a \emph{context}. The dual context $\overline{\Gamma}$ is obtained 
from $\Gamma$ by replacing all the hypotheses of the form $X^\varpos\colon\tau$ by $X^\varneg\colon\tau$ and vice versa, 
and doing the same for lambda variables. We say that $\varphi$ has type $\tau$ in the context $\Gamma$ if the statement 
$\judg{\Gamma}{\varphi\colon\tau}$ can be derived from the rules in Fig.~\ref{fig:typing-rules}.   
A formula without free variables is \emph{well-typed} if the statement $\judg{\emptyset}{\varphi\colon\grtype}$ 
can be derived from these rules. We tacitly assume that each fixpoint variable and each lambda variable is bound at most 
once in a well-typed formula, and that no variable occurs both freely and bound in a formula. Hence, each variable has a 
unique type in the context of a given, well-formed formula. If the type information is clear from
context or not important, we drop it from binders, simply writing $\lambda x.\, \varphi$ and $\mu X.\, \varphi$
for better readability.

A formula is said to be of order $k$ if the maximal order of the type of any subformula in $\varphi$
is $k$. By $\phfl{k}{d}$ we denote the set of well-typed formulas in $\phfl{}{d}$ that are of order $k$. 
Note that $\PHFL$ does indeed constitute an extension of other known formalisms, namely
\begin{itemize}
\item $\phfl{k}{1}$ is the same as $\hfl{k}$ for any $k \ge 0$ and, thus $\phfl{}{1} = \HFL$, and in particular
\item $\phfl{0}{1}$ is just the modal $\mu$-calculus \mucalc while
\item $\phfl{0}{d}$ is the $d$-dimensional polyadic $\mu$-calculus.
\end{itemize}

\paragraph{Semantics.}

Let $\Transsys = (\States, \{\Transition{}{a}{}\}_{a\in\Actions},\Label)$ be an LTS and let $\varphi$
be a well-typed formula.
An \emph{environment} is a function $\eta$ that assigns to each fixpoint variable and each lambda variable
of type $\tau$ an element of $\semm{\tau}{\Transsys}$.

Let $d \ge 1$. The semantics $\sem{\varphi\colon\tau}{\Transsys}{\eta}$ of a $\phfl{}{\dimension}$ formula $\varphi$ of 
type $\tau$, relative to an LTS $\Transsys$ and an environment $\eta$, is an object of $\semm{\tau}{\Transsys}$,
defined recursively as %in Fig.~\ref{fig:phfl-sem}. 
follows.
\begin{align*}
\sem{\judg{\Gamma}{p_i \colon \grtype}}{\Transsys}{\eta} & = \{(\state_1,\dotsc,\state_{\dimension} \in \States^\dimension \mid p \in \Label(\state_i)\} \\
\sem{\judg{\Gamma}{\varphi_1 \vee \varphi_2 \colon \grtype}}{\Transsys}{\eta} & = \sem{\judg{\Gamma}{\varphi_1 \colon \grtype}}{\Transsys}{\eta} \cup
\sem{\judg{\Gamma}{\varphi_2 \colon \grtype}}{\Transsys}{\eta} \\
\sem{\judg{\Gamma}{\neg \varphi \colon \grtype}}{\Transsys}{\eta} & = \States^\dimension \setminus \sem{\judg{\overline{\Gamma}}{\varphi\colon\grtype}}{\Transsys}{\eta} \\
\begin{split}\sem{\judg{\Gamma}{\mudiam{a_i}{\varphi}\colon \grtype}}{\Transsys}{\eta} & = \{(\state_1,\dotsc,\state_\dimension) \in \States^\dimension \mid \text{ ex. } \statealt\text{ s.t. } \\ & \qquad \qquad \qquad (\state_1,\dotsc,\state_{i_1},\statealt,\state_{i+1}, \dotsc, \state_\dimension) \in \sem{\judg{\Gamma}{\varphi \colon \grtype}}{\Transsys}{\eta} \text{ and } \Transition{\state_i}{a}{\statealt}\} \end{split}\\
\sem{\judg{\Gamma}{\subst{\sigma}{\varphi}\colon \grtype}}{\Transsys}{\eta} & = \{(\state_1,\dotsc,\state_\dimension) \in \States^\dimension \mid(\state_{\sigma(1)},\dotsc,\state_{\sigma(\dimension)}) \in \sem{\judg{\Gamma}{\varphi \colon \grtype}}{\Transsys}{\eta}\} \\
\sem{\judg{\Gamma}{\lambda (x^v \colon \tau_2).\ \varphi \colon \tau_2^v \to \tau_1}}{\Transsys}{\eta} &= f \in \semm{\tau_2^v \to \tau_1}{\Transsys} \text{ s.t.~f.a. } y \in \semm{\tau_2}{\Transsys}.\ f(y) 
= \sem{\judg{\Gamma, x^v \colon \tau_2}{\varphi \colon\tau_1}}{\Transsys}{\eta[x \mapsto y]} 
\\
\sem{\judg{\Gamma}{x \colon \tau}}{\Transsys}{\eta} &= \eta(x) \\
\sem{\judg{\Gamma}{(\varphi_1\,\varphi_2) \colon \tau_1}}{\Transsys}{\eta} &= \sem{\judg{\Gamma}{\varphi_1 \colon \tau_2^v \to \tau_1}}{\Transsys}{\eta}\, (\sem{\judg{\Gamma}{\varphi_2 \colon \tau_2}}{\Transsys}{\eta}) \\
\sem{\judg{\Gamma}{\mu (X\colon \tau).\varphi\colon\tau}}{\Transsys}{\eta} &= \bigsqcap_{\tau\to\tau} \{d \in \sem{\tau}{\Transsys}{\eta}\mid \sem{\judg{\Gamma, X\colon\tau^\varpos}{\varphi\colon \tau}}{\Transsys}{\eta[X\mapsto d]} \sqsubseteq_{\tau} d \} \\
\sem{\judg{\Gamma}{X \colon \tau}}{\Transsys}{\eta} &= \eta(X)
\end{align*}
If the type is clear from context, or not important, we simply write $\sem{\varphi}{\Transsys}{\eta}$. We write
$\Transsys, (\state_1,\dotsc,\state_\dimension) \models_\eta \varphi$ if $\varphi\colon\grtype$ and  
$(\state_1,\dotsc,\state_\dimension) \in \sem{\varphi}{\Transsys}{\eta}$. We say that two formulas $\varphi$ and
$\psi$ are \emph{equivalent}, written $\varphi \equiv \psi$, if for all $\Transsys$ and all $\eta$ we have 
$\sem{\varphi}{\Transsys}{\eta} = \sem{\psi}{\Transsys}{\eta}$.

It is well-known that the semantics of HFL and, hence, \PHFL is invariant under $\beta$-reduction and admits the
fixpoint unfolding principle, i.e.\ $\mu X.\, \varphi \equiv \varphi[\mu X.\, \varphi/X]$ where substitution
is defined as usual.

\paragraph{Bisimilarity.} A \emph{bisimulation} $R$ on an LTS 
$\Transsys = (\States, \{\Transition{}{a}{}\}_{a \in \Actions}, \Label)$ is a symmetric relation $R \subseteq \States \times \States$
satisfying the following for all $(s,t) \in R$.
\begin{itemize}
\item $\Label(s) = \Label(t)$,
\item if there is $a \in \Actions$ and $s' \in \States$ s.t.\ $\Transition{s}{a}{s'}$ then there is $t' \in \States$ with
      $\Transition{t}{a}{t'}$ and $(s',t') \in R$,
\item if there is $a \in \Actions$ and $t' \in \States$ s.t.\ $\Transition{t}{a}{t'}$ then there is $s' \in \States$ with
      $\Transition{s}{a}{s'}$ and $(s',t') \in R$.     
\end{itemize} 
Two states $s,t$ are bisimilar, written $s \sim t$, if there is a bisimulation $R$ with $(s,t) \in R$.

Let $d \ge 1$. A set $T \subseteq \States$ is called \emph{bisimulation-invariant} if for all 
$\bar{s} = (s_1,\ldots,s_d),\bar{t} = (t_1,\ldots,t_d) \in \States^d$ such that $s_i \sim t_i$ for all $i \in [d]$, we 
have $\bar{s} \in T$ iff $\bar{t} \in T$. The notion of bisimulation-invariance can straight-forwardly be lifted to 
objects of type $\semm{\tau}{\Transsys}$ for types $\tau \ne \grtype$, cf.\ \cite{DBLP:conf/concur/ViswanathanV04}. 

It is well-known that modal logics cannot distinguish bisimilar models, and it is not surprising that the extensions
beyond pure modal logic that are available in \PHFL do not break this property.

\begin{proposition}[\cite{DBLP:conf/concur/ViswanathanV04,conf/ifipTCS/LangeL14}]
\label{prop:bisiminv}
Let $d \ge 1$, $\Transsys$ be an LTS with state set $\States$ and $\varphi$ be a closed $\phfl{}{d}$ formula of type $\grtype$. Then 
$\semm{\varphi}{\Transsys} \subseteq \States^d$ is bisimulation-invariant.
\end{proposition}

%\paragraph{Examples.} 

\begin{example}[\cite{Otto99}]
\label{ex:bisim}
A standard example shows that bisimilarity itself is definable in $\phfl{0}{2}$, provided that $\Prop$ and 
$\Actions$ are finite. The $\phfl{0}{2}$ formula
\begin{displaymath}
\varphi_\sim := \nu (X: \grtype). (\bigwedge\limits_{p \in \Prop} p_1 \leftrightarrow p_2) \wedge 
(\bigwedge\limits_{a \in \Actions} \mubox{a_1}\mudiam{a_2}X) \wedge \subst{1 \mapsto 2, 2 \mapsto 1}X
\end{displaymath}
is satisfied by a pair $(s,t)$ of some $\Transsys$ iff $s \sim t$. The formula essentially states that $(s,t)$ needs to belong to
the largest set $X$ of states $(s',t')$ that agree on all propositions (first conjunct) and for which $t'$ can match any $a$-transition
out of $s'$, for any $a \in \Actions$, to a pair $(s'',t'') \in X$ (second conjunct). Moreover, $X$ needs to be a symmetric relation (third
conjunct).
\end{example}

To exemplify the use of higher-orderness (here: first-order functions) we can use the definability of finite-trace equivalence
in $\phfl{1}{2}$. 
\begin{example}
\label{ex:traceequiv}
Two states $s,t$ are finite-trace equivalent if whenever there is a sequence 
$\Transition{s}{a_1}{s_1}\Transition{}{a_2}{s_2}\Transition{}{a_3}{}$ $\Transition{\ldots}{a_n}{s_n}$ then there are $t_1,\ldots,t_n$ 
s.t.\ $\Transition{t}{a_1}{t_1}\Transition{}{a_2}{t_2}\Transition{}{a_3}{}\Transition{\ldots}{a_n}{t_n}$ and vice-versa.
\begin{displaymath} 
\varphi_{\mathsf{fte}} := \big(\nu F(x,y). (x \leftrightarrow y) \wedge \bigwedge\limits_{a \in \Actions} F(\mudiam{a_1}x,\mudiam{a_2}y))(\mytrue,\mytrue)
\end{displaymath}
is then satisfied by a pair $(s,t)$ iff $s$ and $t$ are finite-trace equivalent in the sense above. The greatest fixpoint in the formula
expresses an infinite conjunction over all finite paths, these can be thought of to be built step-wise via fixpoint unfolding; the $n$th unfolding
expresses that all finite traces of length $n$ are available in the first component of the tuple (and built via $\mudiam{a_1}$) iff they
are available in the second component (built via $\mudiam{a_2}$.
For better readability we have omitted the types in the formula. The types of $x,y$ are $\grtype$ and that of $F$ is, consequently,
$\grtype^\varboth \to \grtype^\varboth \to \grtype$. 

We remark that it is easily possible to extend the formula to also check for matching atomic propositions along the traces emerging
from $s$ and $t$. Cf.\ \cite{LLVG:TCS:2014} for further examples of how various other process equivalences and preorders can be expressed in
$\phfl{1}{}$.
\end{example}

\subsection{Higher-Order Logic with Least Fixpoints}

We introduce Higher-Order Logic with Least Fixpoints (\HOLFP) to make use of the characterisation of $k$-\ExpTime over the class of ordered structures as the queries definable in order-$(k+1)$ \HOLFP, due to Immerman and Vardi \cite{DBLP:journals/iandc/Immerman86,DBLP:conf/stoc/Vardi82,Imm:lanccc}, resp.\ Freire and Martins \cite{FREIRE201171}.

\paragraph{Types.}
Types for Higher-Order Logic with Least Fixpoints  are constructed, similar to those for \PHFL, from a single base
type and one constructor: $\tau' \Coloneqq \indtype \mid (\tau',\dotsc,\tau')$. Here, however, $\indtype$ is the type of 
\emph{individuals} (like states rather than sets of states), and the tuple type is used to denote (higher-order) relations. 
We define the order\footnote{Note the discrepancy in the traditional ways to assign numerals to orders in the two logics
considered here: order $1$ in \HOLFP refers to, like in ``First-Order Logic'', individual elements and order $2$ is for 
relations like sets thereof. In \PHFL, order $1$ refers to the order of a function, i.e.\ one that takes sets of arguments.
This explains why $\phfl{k}{}$ corresponds to the fragment of \HOLFP of order $k+1$, see also the right column in 
Fig.~\ref{fig:overview}.} of a type via $\ordho{\indtype} = 1$ and 
$\ordho{\tau'_1,\dots,\tau'_n} = 1 + \max \{ \ordho{\tau'_1},\dotsc,\ordho{\tau'_n} \}$.

Let $\Transsys = (\States, \{\Transition{}{a}{}\}_{a\in\Actions},\Label)$ be an LTS. This induces a set-theoretic 
interpretation of types via $\semm{\indtype}{\Transsys} = \States$ and
$\semm{(\tau'_1,\dots,\tau'_n)}{\Transsys} = 2^{\semm{\tau'_1}{\Transsys} \times \dotsb \times \semm{\tau'_n}{\Transsys}}$.

\paragraph{Syntax.}
Let $\hovars = \{X,\dotsc\}$ be a countable set of \emph{higher-order variables}, each implicitly equipped with a type. 
Let $\Prop, \Actions$ be sets
of propositions, resp.\ actions. The syntax of \HOLFP formulas is derived from the following grammar:
\[
\varphi \Coloneqq p(X) \mid a(X,Y) \mid X(Y_1,\dotsc,Y_N) \mid \neg \varphi \mid \varphi \vee \varphi 
\mid \exists (X\colon\tau').\ \varphi \mid \big(\lfp(X,Y_1,\dotsc,Y_N).\ \varphi\big)(Z_1,\dotsc,Z_n)
\]
where $p \in \Prop$, $a \in \Actions$ and $X,Y_1,\dotsc,Y_n,Z_1,\dotsc,Z_n \in \hovars$. Again, dual
operators such as $\wedge$ and $\gfp$ (for greatest fixpoints) are available via the obvious dualities. 
The variable $X$ is bound in $\exists (X\colon\tau').\ \varphi$ and $X,Y_1,\dotsc,Y_N$ are bound in 
$\big(\lfp(X,Y_1,\dotsc,Y_N).\ \varphi\big)(Z_1,\dotsc,Z_n)$.
A formula is \emph{well-formed} if each variable is bound at most once, and, moreover, variables occur
only in a way that matches their type. For example, only a variable of type $\indtype$ can occur in a subformula
of the form $p(X)$, and if a subformula of the form $X(Y_1,\dotsc,Y_n)$ occurs, then $X$ has type $(\tau'_1,\dotsc,\tau'_n)$
for some $\tau'_1,\dotsc,\tau'_n$ and $Y_i$ has type $\tau'_i$ for all $1 \leq i \leq n$. Moreover, in a subformula of the
form $\big(\lfp(X,Y_1,\dotsc,Y_N).\ \varphi\big)(Z_1,\dotsc,Z_n)$, the variable $X$ occurs only under an even number of
negations in $\varphi$. For a more detailed introduction into Higher-Order Logic including formal ways to define 
well-formedness of formulas, cf.\ \cite{vanbenthemdoets:1983a}.
In a well-formed formula each variable has a unique type. An \HOLFP formula $\varphi$ has order $k$ if the order
of the highest type of a variable in $\varphi$ is at most $k$. We write $\holfp{k}$ for the set of \HOLFP formulas
of order at most $k$.

\paragraph{Semantics.}
Let $\Transsys = (\States, \{\Transition{}{a}{}\}_{a\in\Actions},\Label)$ be an LTS.
A variable assignment $\alpha$ is a function that maps each variable in \hovars\ of type $\tau'$ into
$\semm{\tau'}{\Transsys}$. 

Let $\varphi$ be an \HOLFP formula with free variables in $X, Y_1,\dotsc,Y_n$ such that $X$ occurs only under
an even number of negations in $\varphi$ and let $\tau' = (\tau'_1,\dotsc,\tau'_n)$ be the type of $X$ in $\varphi$
while $\tau'_i$ is the type of $Y_i$ in $\varphi$ for $1 \leq i \leq n$. Then, given some variable assignment $\alpha$,
the formula $\varphi$ defines a monotone
function $f \colon \semm{\tau'}{\Transsys} \to \semm{\tau'}{\Transsys}$ via
\[f(M) \mapsto \{(m_1,\dotsc,m_n) \in \sem{\varphi}{\Transsys}{\alpha[Y_1\mapsto m'_1,\dotsc,Y_n \mapsto m'_n]} 
\mid (m'_1,\dotsc,m'_n) \in M\}.\]
By the Knaster-Tarski-Theorem \cite{Kna28,Tars55} this function has a least fixpoint denoted by
$\LFP(X,Y_1,\dotsc,Y_n) \varphi$. Note that we suppress $\Transsys$ and $\alpha$ here since they will always
be clear from context.

The satisfaction relation between an LTS $\Transsys$, a variable assignment $\alpha$ and an \HOLFP formula $\varphi$
is defined inductively as follows. 
\begin{align*}
\Transsys, \alpha \models p(X) &\text{ iff } p \in \Label(\alpha(X)) \\
\Transsys, \alpha \models a(X,Y) &\text{ iff } \Transition{\alpha(X)}{a}{\alpha(Y)} \\
\Transsys, \alpha \models X(Y_1,\dotsc,Y_n) &\text{ iff } (\alpha(Y_1),\dotsc,\alpha(Y_n)) \in \alpha(X) \\
\Transsys, \alpha \models \neg \varphi &\text{ iff } \Transsys, \alpha \not \models \varphi \\
\Transsys, \alpha \models \varphi_1 \vee \varphi_2 &\text{ iff } \Transsys, \alpha \models \varphi_1 \text{ or } \Transsys, \alpha \models \varphi_2 \\
\Transsys, \alpha \models \exists (X\colon\tau') &\text{ iff ex.\ } d \in \semm{\tau'}{\Transsys} \text{ s.t.\ } \Transsys, \alpha[X \mapsto d] \models \varphi \\
\Transsys, \alpha \models \big(\lfp(X,Y_1,\dotsc,Y_N).\ \varphi\big)(Z_1,\dotsc,Z_n) &\text{ iff } (\alpha(Z_1),\dotsc,\alpha(Z_n)) \in \LFP(X,Y_1,\dotsc,Y_n) \varphi.
\end{align*}

\subsection{Descriptive Complexity}

A \emph{query} (of dimension $d$) is a set $Q$ of pairs $(\Transsys,(s_1,\ldots,s_d))$ s.t.\ each $\Transsys$ is finite. It is 
expressed by the \HOLFP formula 
$\varphi$ with free variables $X_1,\dotsc,X_d$, if $Q = \{(\Transsys, (\state_1,\dotsc,\state_d)) \mid \Transsys, [X_1 \mapsto \state_1,\dotsc, X_d\mapsto \state_d] \models \varphi \}$.
Let $k \ge 0$. The complexity class $k$-\ExpTime is defined as $\mathit{DTIME}(2^{n^{\mathcal{O}(1)}}_k)$ where 
$2^m_0 := m$ and $2^m_{k+1} := 2^{2^{m}_k}$. Note that $0$-\ExpTime equals $\PTime$.

For a complexity class $\mathcal{C}$, a $\mathcal{C}$-query is one that can be decided within the resource bounds given by
$\mathcal{C}$. We write $\bisim{\mathcal{C}}$ for the complexity class of $\mathcal{C}$-queries that are bisimulation-invariant. 

We say that a logic $\mathcal{L}$ \emph{captures} a complexity class $\mathcal{C}$ if the model checking problem for $\mathcal{L}$
is in $\mathcal{C}$ and each $\mathcal{C}$-query can be expressed in $\mathcal{L}$. 
 
The Immerman-Vardi Theorem characterises the \PTime-queries 
\emph{over ordered structures} as those expressible in $\holfp{1}$, resp.\ first-order logic with least fixpoints. It's generalisation is the following:

\begin{proposition}[\cite{DBLP:journals/iandc/Immerman86,DBLP:conf/stoc/Vardi82,Imm:lanccc,FREIRE201171}]
\label{prop:IV-gen}
For each $k \ge 0$, $\holfp{k+1}$ captures $k$-\ExpTime over the class of ordered structures.
\end{proposition}

The ordering is only important for the case of $k=0$ as such an ordering can be defined in second-order logic, i.e.\ as soon as
$k \ge 1$. It is an open problem whether a logic exists that captures \PTime over the class of all structures 
(cf.\ e.g.\ \cite{DBLP:conf/lics/Grohe08}). 

The first capturing result for a bisimulation-invariant class is given by Otto's Theorem.

\begin{proposition}[\cite{Otto99}]
\label{prop:otto}
The polyadic modal $\mu$-calculus \polymucalc, or, equivalently $\phfl{0}{}$, captures $\bisim{\PTime}$,
or equivalently, $\bisim{$0$-\ExpTime}$.
\end{proposition}

We give a quick sketch of its proof in order to prepare for the technical developments in Sects.~\ref{sec:quantifier} and
\ref{sec:lower}. The model checking problem for $\phfl{0}{}$ is readily seen to be in \PTime (see also \cite{LL-FICS12}).
The interesting part is to show that every bisimulation-invarariant \PTime query can be expressed 
in $\phfl{0}{}$. In principle, this could be done by encoding runs of polynomially time-bounded Turing machines,
but it is not immediately clear how bisimulation-invariance of the query in question can be used.
Instead, the proof for this rests on a key observation: non-bisimilarity of two states can be expressed in $\phfl{0}{2}$ by 
a least-fixpoint formula, since bisimilarity can be defined via the greatest fixpoint formula in Ex.~\ref{ex:bisim}. Hence, 
non-bisimilarity of two states, as a least fixpoint, has a well-founded reason, i.e.\ one that can be found in finitely many fixpoint unfoldings. 
Ordering the atomic types in an arbitrary way entails a total order on the bisimulation-equivalence classes, and this order can be defined in $\phfl{0}{2}$.
Hence, the LTS in question is ordered, and the Immerman-Vardi Theorem is available, whence the problem reduces
to showing that every query defined by a bisimulation-invariant $\holfp{1}$-query can be expressed equivalently
in $\phfl{0}{}$.

This latter reduction now follows from a rather straightforward translation from $\holfp{1}$ formulas to $\phfl{0}{}$ formulas.
Variables of type $\indtype$ are emulated through polyadicity and variables of type $(\indtype, \dotsc,\indtype)$ are represented
as order-$0$ $\phfl{0}{}$-variables. Since, in the bisimulation-invariant setting, one can always assume that the LTS in question
already is its own bisimulation quotient, bisimilarity and equality coincide. Hence, a subformula of the form $a(X_i,X_j)$ can be 
replaced by the statement that the $i$th component of the relation defined has an $a$-successor that is bisimilar and, hence
equal to the $j$th component. Moreover, since all states in such a bisimulation quotient of a pointed LTS are reachable from 
a distinguished state, existential quantification can be replaced by reachability of a suitable state. It remains to translate 
least fixpoints in $\holfp{1}$ into order-$0$ fixpoints of $\phfl{0}{}$.

Using reasoning along similar lines, Otto's Theorem has also been generalised by one order.

\begin{proposition}[\cite{conf/ifipTCS/LangeL14}]
\label{prop:capturebisimexptime}
$\phfl{1}{}$ captures $\bisim{\ExpTime}$.
\end{proposition}

\section{Upper Bounds}
\label{sec:upper}
% !TEX root = main.tex

Capturing a complexity class $\mathcal{C}$, defined by some restricted resource consumption, by a logic $\mathcal{L}$ 
contains two parts: what is commonly seen as the lower bound consists of showing that every query which can be evaluated
in complexity $\mathcal{C}$ can also be defined in the logic $\mathcal{L}$. The upper bound is established by showing
the contrary. This is relatively easy as it suffices to show that queries definable in $\mathcal{L}$ can be evaluated 
in complexity $\mathcal{C}$, in other words that the model checking problem for $\mathcal{L}$ belongs to class $\mathcal{C}$.

Here we do this for the fragments of \PHFL of arbitrary but fixed arity $d$ and arbitrary order $k$, w.r.t.\ the classes
$k$-\ExpTime of the exponential time hierarchy. We do so by extending the reduction of the model checking problem for a 
polyadic logic to that of its monadic fragment \cite{LL-FICS12}. Note that $\phfl{k}{1}$ equals $\hfl{k}$ -- the fragment 
of (non-polyadic) Higher-Order Fixpoint Logic of formulas of type order at most $k$. The complexity of model checking such 
fragments is known: 

\begin{proposition}[\cite{als-mchfl07}]
\label{prop:hflmodelcheck}
Let $k \ge 1$. The model checking problem for $\hfl{k}$ is $k$-\ExpTime-complete.
\end{proposition} 
The corresponding result for $\phfl{k}{d}$, first stated without proof in \cite{conf/ifipTCS/LangeL14}, follows
via a reduction:
\begin{theorem}
\label{thm:upper}
Let $k,d \ge 1$. The model checking problem for $\phfl{k}{d}$ is in $k$-\ExpTime.
\end{theorem}

\begin{proof}
By a polynomial reduction to the model checking problem for \HFL. Let $k,d \ge 1$ and an LTS 
$\Transsys = (\States, \{\Transition{}{a}{}\}_{a\in\Actions},\Label)$ over $\Actions$ and $\Prop$ be given. We construct
its $d$-product $\Transsys^d$ over the action set 
$\Actions^d := \{a_i \mid a \in \Actions, i \in [d] \} \cup \{ \sigma \mid \sigma: [d] \to [d] \}$ and atomic propositions
$\Prop^d := \{ q_i \mid q \in \Prop, i \in [d] \}$ as $(\States^d, \{\Transition{}{x}{}\}_{x \in \Actions'_d}, \Label')$ 
where, for all $q \in \Prop$, $i \in [d]$, $a \in \Actions$, $s_1,\ldots,s_d,t_1,\ldots,t_d \in \States$ we have
\begin{itemize}
\item $q_i \in \Label'(s_1,\ldots,s_d)$ iff $q \in \Label(s_i)$,
\item $\Transition{(s_1,\ldots,s_d)}{a_i}{(t_1,\ldots,t_d)}$ iff $\Transition{s_i}{a}{t_i}$ and $t_j=s_j$ for all $j \ne i$,
\item $\Transition{(s_1,\ldots,s_d)}{\sigma}{(t_1,\ldots,t_d)}$ iff $t_j = \sigma(s_j)$ for all $j \in [d]$.
\end{itemize}
Next, we translate a $\phfl{k}{d}$ formula $\varphi$ inductively into a $\phfl{}{1}$ formula $\widehat{\varphi}$ as follows.
The operation $\widehat{\cdot}$ acts homomorphically on all operators apart from the following three cases.
\begin{displaymath}
\widehat{\subst{\sigma}\psi} := \mudiam{\sigma}{\widehat{\psi}} \enspace, \quad 
\widehat{\mudiam{a_i}\psi} := \mudiam{a_i}{\widehat{\psi}} \enspace, \quad 
\widehat{p_i} := p_i \ . 
\end{displaymath}
The latter two cases may be confusing as $\widehat{\cdot}$ seems to not change those operators either. However, in 
e.g.\ the second case, $\mudiam{a_i}$ on the left side is a polyadic modality combining the action $a \in \Actions$ with 
the index $i \in [d]$. On the right side, $\mudiam{a_i}$ is a monadic modality over the action $a_i \in \Actions^d$. Likewise
in the third case.

This equality in syntax for two technically different modal operators, resp.\ atomic formulas is intended because of the
following connection: the $d$-tuple $(s_1,\ldots,s_d)$ of states in $\Transsys$ satisfies the $\phfl{}{d}$ formula 
$\mudiam{a_i}\psi$ iff the state $(s_1,\ldots,s_d)$ of $\Transsys^d$ satisfies the $\phfl{}{1}$ formula 
$\mudiam{a_i}{\widehat{\psi}}$. A similar statement can be made for atomic propositions and these can easily be generalised
to show by induction on the syntax of \PHFL that for all $\phfl{}{d}$ formulas $\varphi$, all 
$s_1,\ldots,s_d \in \States$ and all environments $\eta$ we have: $(s_1,\ldots,s_d) \in \sem{\varphi}{\Transsys}{\eta}$ 
iff $(s_1,\ldots,s_d) \in \sem{\widehat{\varphi}}{\Transsys^d}{\eta}$. 

This establishes correctness of the reduction. Note that $\widehat{\varphi}$ is a $\phfl{k}{1}$-, i.e.\ $\hfl{k}$ formula
whenever $\varphi \in \phfl{k}{d}$. Moreover, both $\Transsys^d$ and $\widehat{\varphi}$ are easily seen to be constructible
in polynomial time (for fixed $d$). Thus, by Prop.~\ref{prop:hflmodelcheck} model checking $\phfl{k}{d}$ is also in $k$-\ExpTime.
\end{proof}

\section{Higher-Order Quantification}
\label{sec:quantifier}
% !TEX root = main.tex

To show that every bisimulation-invariant $k$-\ExpTime query can be expressed in $\phfl{k}{}$ it suffices, due to Prop.~\ref{prop:IV-gen}, 
to show that every $\holfp{k+1}$-query can be translated into a $\phfl{k}{}$ formula. The main challenge here is to deal with existential quantification, which has no obvious equivalent 
in \PHFL. In \cite{Otto99}, first-order existential quantification is replaced by reachability of a suitable state, 
which is sufficient in the bisimulation-invariant setting. Higher-order quantification does not have such an
obvious correspondent - there is no notion of a set, or a set of sets, etc.\ being reachable. Moreover, \PHFL
does only have a type for sets of tuples of states, not for sets of sets etc. 
We solve this problem by replacing higher-order types  by a variant of their characteristic function, something
that fits quite naturally into the \PHFL world. We then lift the order on the states inherited from \cite{Otto99} 
to sets of tuples of states and then to said characteristic functions by ordering them lexicographically. 
We can enumerate sets, functions and so on alongside this order, and hence, we can mimic existential quantification
over some \HOLFP type by an enumeration of corresponding characteristic functions. 

The increased complexity of this approach compared to simple reachability requires two adaptations: first, since we use the
order of sets, functions etc.\ present in the bisimulation-invariant setting, we often have to compare two such objects
w.r.t.\ this order. Comparing e.g.\ two sets of $\dimension$-ary tuples, however, requires a formula containing types of width $2d$. 
Also, \HOLFP-formulas can define a query of some width, yet contain types, resp.\ quantification over objects of much higher
width. Hence, in order to keep the presentation simple, the formulas we develop subsequently will be of some unspecified,
yet generally quite high arity, i.e.\ they will be in $\phfl{k}{\dimension}$ for some $\dimension$ that is large compared with the width of
the original query. The exact value of $\dimension$ will be given towards the end of the translation.

Since we have already agreed to blow up the width used in our translated formulas, we can also make things easier by reserving
certain positions in the tuples we work with for special tasks. We generally use the last two positions in our tuples (those
with indices $\dimension$ and $\dimension-1$) to compare individual states w.r.t.\ the order from \cite{Otto99}, and we will use the next 
$r$ positions from the right, i.e.\ those with indices $\dimension-r-1,\dotsc, \dimension-2$ for some $r$, to keep copies of states such that the 
whole LTS is reachable from at least one of these states, in order to keep the pattern for first-order quantification valid. This will
be made formal after we revisit the pattern for existential quantification just below. We then subsequently expand quantification towards sets 
and characteristic functions. The final translation is then given in Sec.~\ref{sec:lower}.

\paragraph{Reachable States.}

In \cite{Otto99}, existential quantification of first-order logic was replaced by reachability of a suitable state in the LTS in question,
using the pattern given subsequently. First, we recall the role of the substitution operator:  
Let $\sigma^{i \leftarrow j}$ be defined by $\sigma^{i \leftarrow j}(i) = j$ and $\sigma^{i \leftarrow j}(i') = i'$ if $i' \not = i$. Then
$(\overline{\state} \in \sem{\subst{\sigma^{i \leftarrow j}}{\varphi}}{\Transsys}{\eta}$ iff 
$\overline{\state}[\state_j/i] = (\state_1,\dotsc,\state_{i-1},\state_j,\state_{i+1},\dotsc,\state_\dimension) \in \sem{\varphi}{\Transsys}{\eta}$
for all $\Transsys$ and $\eta$.

Now let $\Actions$ be a set of actions, let $\dimension > 2$, $r \leq \dimension-2$ and 
$i \leq \dimension - r - 2$. Recall that, for the time being, we assume that every tuple we work with is such that all states in an LTS
are reachable from one of the states at indices $\dimension-r-1,\dotsc,\dimension-2$ of the tuple and that we reserve the last two positions
for comparisons (see below).  Consider the formula
\[
\exists_i \varphi \coloneqq \bigvee_{j = \dimension -r -1}^{\dimension-2} \subst{\sigma^{i\leftarrow j}}{\big(\mu (X\colon\grtype).\ \varphi \vee \bigvee_{a \in \Actions} \mudiam{a_i}{X}\big)}.
\]
for some $\varphi \in \phfl{}{\dimension}$. We have the following:
\begin{observation}
\label{obs:ind-quant}
Let $\Transsys$ be an LTS over $\Actions$ and let $\state = (\state_1,\dotsc,\state_\dimension)$ such that all states in $\Transsys$ 
are reachable from at least one state in
 $\state_{\dimension-r-1},\dotsc,\state_{\dimension-2}$. Then $\Transsys, \overline{\state} \models \exists_i \varphi$ iff 
there is a state $\statealt$ in $\Transsys$ such that $\Transsys, \overline{\state}[\statealt/i] \models \varphi$. 
\end{observation}
As said before, quantification over higher-order types is more complicated, but  we can replace reachability by enumeration in 
lexicographical order for higher-order types. Towards this, note that for all 
$\dimension \geq 2$ there is some $\phfl{0}{2}$-formula $\varphi_<$ defining a transitive and
irreflexive relation $<$ such that, for all LTS $\Transsys$ and $d$-tuples $\overline{\state} = (\state_1,\dotsc,\state_d)$ we have 
that $\Transsys, \overline{\state} \models \varphi_<$ iff $\state_{\dimension-1} < \state_\dimension$. This formula is defined in \cite{Otto99} using
a variant of the negation of the formula from Ex.~\ref{ex:bisim}. The actual position of the tuple elements that are compared is not important, we choose to fix it here for consistency.

A crucial ingredient for the correctness of the quantification pattern above is that every state in the LTS is reachable from 
the states in positions $\dimension-r-1,\dotsc,\dimension-2$. For the first-order case, this can be guaranteed by never manipulating
the components with the respective indices. In the higher-order setting, this does not suffice as one deals with arbitrary sets, functions, etc.
Hence, for the remainder of the section all formulas are assumed to have a free lambda variable $e$ of type $\grtype$ 
that is as follows:
\begin{definition}
\label{def:good}
Let $\Transsys$ be an LTS and let $r$ be fixed. Then an interpretation $\eta$ is \emph{good} if $\eta(e)$ is a set of the form $M\times\{\state_{\dimension-r-1}\}\times\dotsb\times\{\state_{\dimension-2}\}\times\States^2$ such that 
$\emptyset \not = M \subseteq \States^{\dimension-2-r}$ and each state of $\Transsys$ is reachable from one of the 
$\state_{\dimension-r-1},\dotsc,\state_{\dimension-2}$. 
\end{definition}
We do not make this free variable explicit, since it is always assumed to be there. We will see in Sect.~\ref{sec:lower} 
how this intended interpretation can be enforced. This stipulation formalises the informal idea given above. Note that
in \cite{Otto99}, it was assumed that all states of the LTS in question are reachable from a singular state, but this is not
a necessary requirement for the argument to work.

Finally, let $w$ and $\dimension$ be such that $2w+r+2 \leq \dimension$.
The intuition here is that, in order to translate from $\HOLFP$, we will have to deal with sets and higher-order sets of arity at most $w$.

\paragraph{Quantification for Sets.}

We now define a similar pattern to that for the first-order case which allows us to iterate over all sets of $w$-tuples in an LTS. 
This is not the same as enumerating $\semm{\grtype}{\Transsys}$ since $w < \dimension$.

Let $\sigma_i$ be defined via $\sigma_i(\dimension-1) = i$, $\sigma_i(\dimension) = i+w$ and $\sigma_i(j) = j$
if $j < \dimension-1$.  Then $\overline{\state} \in \sem{\subst{\sigma_i}{\varphi}}{\Transsys}{\eta}$ iff
$\overline{\state}[\state_i/\dimension-1, \state_{i+w}/\dimension] \in \sem{\varphi}{\Transsys}{\eta}$. The intended use for
this substitution is to compare the elements at indices $i$ and $i+w$ w.r.t.\ to the order induced by $\varphi_<$. Remember that this
formula always compares the last two elements of the tuple. Moreover, let $\sigma_{\rightarrow w}$ be defined by 
$\sigma_{\rightarrow w}(i) = w+i$ for $i \leq w$ and $\sigma_{\rightarrow w}(j) = j$ for $j > w$. 
Then $(\state_1,\dotsc,\state_\dimension) \in \sem{\subst{\sigma_{\rightarrow w}}{\varphi}}{\Transsys}{\eta}$ iff
$(\state_1,\dotsc,\state_w,\state_1,\dotsc,\state_w,\state_{2w+1},\dotsc,\state_\dimension) \in \sem{\varphi}{\Transsys}{\eta}$. 
The intended use here is to shift the first $w$ elements of a formula to the right to make room for another tuple at the first $w$
positions such that the two tuples can be compared lexicographically.

Consider the formula $\exists^{(w)} x.\varphi \coloneqq \big(\mu (F\colon \grtype \to \grtype).\, \lambda (x\colon\grtype).\, \varphi \vee F(\nxt^{(w)}(x))\big) \myfalse$ where 
\begin{align*}
\varphi_<^w          &\coloneqq \bigvee_{i=1}^w \subst{\sigma_i}{\varphi_<} \wedge \bigwedge_{j=1}^{i-1} \subst{\sigma_j}{\neg \varphi_<} \\
\varphi_<^{(w)}(x,y) &\coloneqq \exists_{1}\dotsc\exists_{w}.\ y \wedge \neg x
                        \wedge \subst{\sigma_{\rightarrow w}}{\big(\forall_{1}\dotsc\forall_{w}. \varphi_<^w \rightarrow x \rightarrow y\big)}. \\
\nxt^{(w)}(x)        &\coloneqq \lambda(x \colon \grtype).\ e \wedge \neg x \wedge \subst{\sigma_{\rightarrow w}}{\big(\forall_{1}\dotsc\forall_{w}.\ \varphi_<^{w} \rightarrow x\big)} \\
                     & \qquad \qquad \;\;\;\; \vee e\wedge x \wedge \subst{\sigma_{\rightarrow w}}{\big(\exists_1\dotsc\exists_w.\ \varphi_<^{w} \wedge \neg x\big)} 
\end{align*}
The first formula $\varphi_<^w$ implements lexicographical comparison of the first $w$ elements in a tuple to the second $w$ elements. The second formula $\varphi_<^{(w)}$ lifts the order from tuples to sets of tuples via the lexicographical order induced by the order on the tuples.
Finally, the formula $\nxt^{(w)}$ is a predicate transformer that consumes a set of tuples. It returns a set of tuples that is the lexicographical successor of the input in the order induced by the order on the first $w$ elements of the individual tuples: the output contains a tuple iff either the input does contain it, but not all lexicographically smaller tuples, or if the input does not contain it, but all lexicographically smaller tuples. Also note the role of $e$ that filters out
all tuples that do not adhere to our stipulation that the whole LTS be reachable from one of the states at indices $\dimension-r-1,\dotsc,\dimension-2$.
\begin{lemma}
\label{lem:quantifier-sets}
Let $\Transsys$ be an LTS with state set $\States$ and let $\eta$ be good. Then $\Transsys,\overline{\state} \models_\eta \exists^{(w)} x.\, \varphi$ 
iff there is $M \subseteq \States^w$ such that $\Transsys, \overline{\state} \models_{\eta[x\mapsto M \times \States^{\dimension - w} \cap \eta(e)]} \varphi$.
\end{lemma}
The proof consists of verifying the informal intuition above.
We write $\forall^{(w)} x.\ \varphi$ to denote $\neg \exists^{(w)} x.\ \neg \varphi$. We write $\exists^{(w)} x_1,\dotsc,x_n.\ \varphi$
for $\exists^{(w)} x_1 \dotsb \exists^{(w)} x_n.\ \varphi$, and similarly for $\forall^{(w)} x.\ \varphi$.

\paragraph{Generalised Higher-Order Quantification.}

We have just seen 
how existential quantification can be emulated
for individual states in an LTS and, with some restrictions, for sets of $w$-tuples of an LTS.  
For other types that commonly appear in \HOLFP, i.e.\ relations of higher order, there is no
immediate \PHFL equivalent, since all types beyond $\grtype$ are function types. However, we can 
use these function types to emulate the \HOLFP types to a sufficient degree. For the sake of simplicity,
we only consider types of a special form; we argue in Sect.~\ref{sec:lower} why this is not a restriction.

Let $\tau_{w,k}$ be inductively defined via $\tau_{w,0} = \grtype$ and $\tau_{w,i+1} = \tau_{w,i} \to \dotsb \to \tau_{w,i} \to \grtype$
where $\tau_{w,i}$ is repeated $w$ many times. Given $\Transsys$, let $\semm{\tau^\circ_{w,i}}{\Transsys} = \semm{\tau_{w,i}}{\Transsys}$
for $i \leq 1$ and let 
\[
\semm{\tau^\circ_{w,k}}{\Transsys} = 
\{f \in \semm{\tau_{w,k}}{\Transsys} \mid f(f_1,\dotsc,f_w) = \States \text{ or } f(f_1,\dotsc,f_w) = \emptyset 
\text{ f.a.\ } f_1,\dotsc,f_w \in \semm{\tau^\circ_{w,k-1}}{\Transsys} \}\]
for $k \geq 2$. The important distinction here is that $\semm{\tau^\circ_{w,k}}{\Transsys}$ is the 
restriction of $\semm{\tau_{w,k}}{\Transsys}$ to those functions that always return  either the
full set of states or the empty set, at least on inputs from $\semm{\tau^\circ_{w,k-1}}{\Transsys}$. This 
is desirable since we want to use functions in $\semm{\tau^\circ_{w,k}}{\Transsys}$ to emulate higher-order variables
of a special form. Given $x, x_1,\dotsc,x_w$ of the appropriate type, the question whether 
$\overline{s} \in \sem{x(x_1,\dotsc,x_w)}{\Transsys}{\eta}$ 
does not depend on $\overline{s}$ (as in modal logics), but is uniform over the while LTS. However,
since $\semm{\tau_{w,k}}{\Transsys}$ also contains functions that are not uniform starting from \PHFL-order $2$, 
we restrict ourselves to functions that are uniform on the necessary inputs (i.e.\ those that are themselves sufficiently uniform).

Consider the following formulas for $k \geq 1$, where $\overline{x} = x_1,\dotsc,x_w$ and 
$\overline{y} = y_1,\dotsc,y_w$:
\begin{align*}
\varphi_<^{w,k-1}(x_1,\dotsc,x_w,y_1,\dotsc,y_w) &\coloneqq \bigvee_{i = 1}^w \varphi_<^{(w),k-1}(x_i, y_i)
                                                           \wedge \bigwedge_{j = 1}^{i-1} \neg \varphi_<^{(w),k-1}(x_j,y_j) \\
\varphi_<^{(w),k}(x,y)  &\coloneqq \exists^{w,k-1} \overline{x}.\ y(\overline{x}) \wedge \neg x(\overline{x}) \\
                                    &\; \wedge \forall^{w,k-1} \overline{y}.\ \varphi_<^{w,l}(\overline{y},\overline{x}) \rightarrow
																						    x(\overline{y}) \rightarrow y(\overline{y}) \\
\myfalse_{(w),k} & \coloneqq \lambda (\overline{x}\colon\tau_{w,k-1}).\, \myfalse \\
\nxt^{w,k}(x) &\coloneqq \lambda (x\colon\tau_{w,k}).\  \lambda (\overline{x}\colon\tau_{w,k-1}).\  \\
& \big( \neg x(\overline{x}) \wedge \forall^{w,k} \overline{y}.\ 
\varphi_<^{w,k}(\overline{y},\overline{x}) \rightarrow x(\overline{y})\big) \\
& \vee \big(  x(\overline{x}) \wedge \exists^{w,k} \overline{y}.\ 
\varphi_<^{w,k}(\overline{y},\overline{x}) \wedge \neg x(\overline{y})\big) \\
\exists^{w,k}(x).\ \varphi &\coloneqq \big( \mu (F\colon \tau_{w,k} \to \tau_{w,k}).\ \lambda (x\colon\tau_{w,k}).\, \varphi \vee F(\nxt^{w,k} x)\big)\, \myfalse_{(w),k}                                          
\end{align*}
where $\varphi_<^{(w),k-1} =  \varphi_<^{(w)}$ in case $k = 1$ and $\exists^{w,k-1} x. \varphi = \exists^{(w)} x. \varphi$ if $k=1$. 

Similarly as in the definitions given before Lemma~\ref{lem:quantifier-sets}, these formulas lift quantification up by one level on 
the type hierarchy. The formula $\varphi_<^{w,k-1}$ compares width-$w$-tuples of functions of type $\tau^\circ_{w,k-1}$ lexicographically
using the previously defined formula $\varphi_<^{(w),k-1}$ that compares individual such functions. The formula $\varphi_<^{(w),k}$ then
lifts this to individual functions of the next type, using existential and universal quantification. The formula $\nxt^{w,k}$ again consumes
a function of type $\tau^\circ_{w,k}$ and returns the lexicographically next one using the standard definition of
binary incrementation, while $\exists^{w,k}$ implements existential quantification for $\tau^\circ_{w,k}$ by iterating
through all possible candidates using $\nxt^{w,k}$.

\begin{lemma}
\label{lem:quantifier-ho}
Let $\Transsys$ be an LTS and let $\eta$ be good. Then $\Transsys,\overline{\state} \models_\eta \exists^{w,k} x.\, \varphi$ 
iff there is $f \in \semm{\tau^\circ_{w,k}}{\Transsys}$ such that $\Transsys, \overline{\state} \models_{\eta[x\mapsto f]} \varphi$.
\end{lemma}
The proof follows the same pattern as that of Lemma~\ref{lem:quantifier-sets} by verifying that the individual formulas 
do what is claimed above.

Lemmas~\ref{lem:quantifier-sets} and \ref{lem:quantifier-ho} justify the use of $\HOLFP$-style quantification symbols for individual
states, sets of type $M \times \States^{\dimension-w}\cap\eta(e)$ for $M \subseteq \States^s$ and for $\tau^\circ_{w,k}$ for all $k \geq 2$. 
Note that the latter do neither coincide with \HOLFP types nor with the respective $\tau_{w,k}$.

\section{Lower Bounds}
\label{sec:lower}
% !TEX root = main.tex

\paragraph{Homogeneous Types.}

In order to simplify the translation from \HOLFP to \PHFL, we restrict the set of types 
that can be used in \HOLFP formulas. Let $w\geq 2$ be fixed but arbitrary. 
Define $\tau'_{w,k}$ as $\tau'_{w,1} = \indtype$, $\tau'_{w,i+1} = (\tau'_{w,i},\dotsc,\tau'_{w,i})$ with
$w$ many repetitions of $\tau'_{w,i}$. 
\begin{lemma}
\label{lem:homo}
If $\varphi \in \holfp{k}$ defines a query $Q$, then there is $\varphi' \in \holfp{k}$ that defines 
the same query, but the only types used in $\varphi'$ are $\tau'_{w,0},\dotsc,\tau'_{w,k}$ for some $w$.
\end{lemma}
\begin{proof}
There are two principles that are used here: If $\varphi$ contains a type of the form $\tau'' = (\tau',\dotsc,\tau')$
with $w' \leq w$ many repetitions of $w$, we can replace $\tau''$ by $(\tau',\dotsc,\tau')$ with exactly
$w$ many repetitions of $\tau'$ by changing the respective type everywhere in the formula and requiring
at quantifiers that, e.g., the last $w-w'$ components are equal to the $w-w'-1$st in every tuple contained in a set.
Hence, every type can be assumed to have width exactly $w$.

It remains to deal with inhomogeneous types, e.g.\ those of the form $(\tau',\tau'',\dotsc)$ such that $\ordho{\tau'} \not = \ordho{\tau''}$.
This can be remedied by increasing the type of the lower order by one order, and requiring at quantification
steps that the only tuple in the set be a singleton of the form $(M,\dotsc,M)$ of width $w$. This procedure
needs to be chained if the orders of constituent types diverge by more than one.
\end{proof}

\paragraph{Type Correspondence.}
Now that we can assume w.l.o.g.\ that all $\holfp{k}$-definable queries are defined by an $\holfp{k}$ formula which
uses only the types $\tau'_{w,i}$ for $i \leq k$ we can observe that these types are quite similar to the types 
$\tau_{w,k}$ defined in the previous section. 
%Sec.~\ref{sec:gen-quant}. 
In fact,
we want to emulate the type $\tau'_{w,k}$ with $k>1$ by $\tau^\circ_{w,k-2}$. The type $\tau'_{s,1}$ will
be handled by polyadicity as in \cite{Otto99}.\footnote{It would also be possible to completely
eliminate this type from $\holfp{k}$ for $k \geq 2$; cf.\ similar constructions in the context of MSO 
and automata theory.}

Let $\Transsys$ be an LTS with state set $\States$ and let $\dimension \geq w \geq 2$. For each $k \geq 2$, we define a translation
$\SEMTRANS_k^{\Transsys}\colon\semm{\tau'_{w,k}}{\Transsys} \to \semm{\tau^\circ_{w,k-2}}{\Transsys}$ via 
$\semtrans{M}{\Transsys}{2} = M \times \States^{\dimension-w}$ and  
$\semtrans{M}{\Transsys}{i+1} = f \in \semm{\tau_{w,k-2}}{\Transsys}$ s.t.\ 
\begin{align*}
f(f_1,\dotsc,f_n) = \left\{
\begin{aligned}
\States, &\text{ if } f_j = \semtrans{x_i}{\Transsys}{i} \text{ f.a. } 0 \leq j \leq w \text{ and } (x_1,\dotsc,x_w) \in M \\
\emptyset, &\text{ otherwise.}
\end{aligned}\right.
\end{align*}
Hence, a variable assignment $\alpha$ such that all variables are of types $\tau'_{w,2},\dotsc,\tau'_{w,k}$ induces an
environment $\eta_\alpha$ with definitions for all those variables of order $\geq 2$ via $\eta_\alpha(X) = \semtrans{X}{\Transsys}{i}$
where $i$ is the order of $X$.

\paragraph{The Translation.}

We are now ready to extend the translation given in \cite{Otto99} to $\holfp{k}$ with $k \geq 2$. Towards this,
we present a syntactical translation $\SYNTRANS$ from $\holfp{k}$ with $k \geq 2$ into $\phfl{\dimension}{k-1}$ 
for some $\dimension \geq 2$.
\begin{lemma}
\label{lem:translate}
Let $\varphi \in \holfp{k}$ be bisimulation invariant and have free first-order variables $X_1,\dotsc,X_r$ and, hence define a query of width $r$. 
Moreover, let $\tau'_{w,0},\dotsc,\tau'_{w,k}$ be the only types in $\varphi$. W.l.o.g.\ let $2w \geq r$. Let
$\dimension = 2w+r+2$. Then there is $\varphi' \in \phfl{\dimension}{k-1}$ such that, for all LTS
$\Transsys$ that are a bisimulation quotient, we have 
$\Transsys, \alpha \models \varphi$ iff $\Transsys, (\alpha(X_1),\dotsc,\alpha(X_r)) \models_{\eta_\alpha[e\mapsto M]} \varphi'$
where $\emptyset \not = M' \subseteq \States^r$ and $M = M'\times \States^{\dimension-2r-2}\times\{\alpha(X_1)\}\times\dotsb\times\{\alpha(X_r)\}\times\States^2$.
\end{lemma}
\begin{proof}
Let $\SYNTRANS$ be given via
\begin{align*}
\syntrans{p(X_i)} &\coloneqq p_i \\
\syntrans{a(X_i,X_j)} &\coloneqq \mudiam{a_i}{\subst{\sigma_{i,j}}{\varphi_\sim}} \\
\syntrans{\varphi_1 \vee \varphi_2} &\coloneqq \syntrans{\varphi_1} \vee \syntrans{\varphi_2} \\
\syntrans{\neg \varphi} &\coloneqq \neg \syntrans{\varphi} \\
\syntrans{\exists (X_i\colon\indtype).\ \varphi} &\coloneqq \exists_i.\ \syntrans{\varphi} \\
\syntrans{\exists(X\colon \tau'_{w,1}).\, \varphi} &\coloneqq \exists^{(w)}(x).\  \syntrans{\varphi} \\
\syntrans{\exists(X\colon \tau'_{w,k}).\ \varphi} &\coloneqq \exists^{w,k}(x).\ \syntrans{\varphi} \text{ if } k \geq 2 \\
\syntrans{X(X_{i_1},\dotsc,X_{i_w})} &\coloneqq \subst{\sigma}x \text { if } (X\colon\tau'_{w,1}) \\
\syntrans{X(X_1,\dotsc,X_w)} &\coloneqq x(x_1,\dotsc,x_w) \text{ if } (X \colon \tau'_{w,k}) \text{ and } k \geq 2 \\
\syntrans{\big(\lfp(X,Y_{1},\dotsc,Y_{w}).\ \varphi\big)(Z_{i_1},\dotsc,Z_{i_w})} &\coloneqq \subst{\sigma}{\mu (F\colon\grtype\to\grtype).\  \syntrans{\varphi}} \text{ if } (X\colon\tau_{w,2}) \\
\syntrans{\big(\lfp(X,Y_{1},\dotsc,Y_{w}).\ \varphi\big)(Z_{1},\dotsc,Z_{w})} &\coloneqq \big(\mu (F\colon\tau_{w,k-2}\to\tau_{w,k-2}).\ \lambda(y_1,\dotsc,y_w\colon\tau_{w,k-3}).\ \\& \qquad \qquad \qquad \quad \syntrans{\varphi}\big)(Z_1,\dotsc,Z_w)  \text{ if } (X\colon\tau_{w,k}) \text{ and } k \geq 3
\end{align*}
where $\sigma_{i,j}$ is defined via $\sigma_{i,j}(\dimension-1) = i, \sigma_{i,j}(\dimension) = j$ and $\sigma_{i,j}(k) = k$ 
for $k < \dimension-1$, and $\varphi_\sim$ is the formula from Ex.~\ref{ex:bisim}, resp.\ \cite{Otto99} that holds iff the states in positions
$\dimension-1$ and $\dimension$ are bisimilar, and, finally $\sigma$ is defined via $\sigma(j) = i_j$ for $0 \leq j \leq w$ and
$\sigma(k) = k$ for $k \geq w$. The substitions $\sigma_{i,j}$, resp.\ $\sigma$ serve to swap the elements $i$ and $j$ to positions
$\dimension-1$ and $\dimension$, resp.\ to reorder the elemtents according to the variable order used on the left
side of the translation.¸

The claim now follows by induction over the syntax tree of $\varphi$. 
 The first five cases are as in \cite{Otto99}. The next four cases
are by Lemmas~\ref{lem:quantifier-sets} and \ref{lem:quantifier-ho}, resp.\ the definition of $\eta_\alpha$.
The fixpoint cases are by an induction over the individual stages of the fixpoint iteration; note that
the Knaster-Tarski-based semantics of the fixpoints in \HOLFP and \PHFL can be equivalently replaced by 
semantics based on the Kleene Fixpoint Theorem \cite{Kleene:1938}. This mirrors the proof of the fixpoint
case in \cite{Otto99}.
\end{proof}

\begin{theorem}
\label{thm:capture}
$\phfl{k}{}$ captures \bisim{k-\ExpTime} for all $k \geq 0$.
\end{theorem}
\begin{proof}
The cases of $k=0$ and $k=1$ are of course already known \cite{Otto99,conf/ifipTCS/LangeL14}. For $k \ge 2$, the upper 
bound is shown in Thm.~\ref{thm:upper}. For the lower bound, it suffices, due to Prop.~\ref{prop:IV-gen}, 
to show that for any bisimulation-invariant query defined by an \holfp{k+1} formula there is an equivalent 
formula in $\phfl{k}{}$. Let $\varphi$ be a formula 
defining such a query $Q$ of width $r \geq 1$, and w.l.o.g.\ $\varphi$ contains only the types $\tau_{w,0},\dotsc,\tau_{w,k}$.

Let $\varphi'$ be the formula obtained from the translation in Lemma~\ref{lem:translate} and let $\sigma$
be defined via $\sigma(\dimension-r-2+i) = i$ for $1 \leq i \leq r$, and $\sigma(j) = j$ for $j \leq
\dimension - r -2$ or $j \geq \dimension-1$. Note that $\sigma$ simply copies the first $r$ many
states in a tuple to the positions $\dimension-r-1,\dotsc,\dimension-2$ where they serve to retain goodness (cf.\ Def.~\ref{def:good}).
Then
$\psi = (\lambda (e\colon\grtype)\varphi)\subst{\sigma}{\mytrue}$ defines the 
query $\{(\Transsys, (\state_1,\dotsc,\state_\dimension)) \mid (\Transsys, (\state_1,\dotsc,\state_r)) \in Q\}$.
Towards this, note that $\sem{\subst{\sigma}{\mytrue}}{\Transsys}{} = \{(\state_1,\dotsc,\state_\dimension) \mid
\state_i = \state_{\dimension-r-2+i} \text{ for } 1 \leq i \leq r\} = M$.
Hence, $\sem{\psi}{\Transsys}{} = \sem{\varphi'}{\Transsys}{\eta}$ where $\eta = \emptyset[e\mapsto M]$.
Since $\eta$ is good, Lemma~\ref{lem:translate} is applicable.

What remains is to argue that a query that, on each LTS $\Transsys$, returns all tuples in $M \times \States^{\dimension-r}$
where $M \subseteq \States^r$, is in fact a query that defines tuples of width $r$. This can either be done formally
by expanding the semantics of the substitution operator $\subst{\sigma}{}$ to return tuples in the width of its co-domain,
or by simply considering queries of the above form to be of width $r$, cf.~\cite{Otto99}.
\end{proof}

\section{Conclusion}
\label{sec:concl}
% !TEX root = main.tex

We have extended the descriptive complexity of bisimulation-invariant queries, as started by
Otto \cite{Otto99} with the capturing of the class \bisim{\PTime} with the logic \polymucalc, to further
time complexity classes, namely the classes \bisim{$k$-\ExpTime} of bisimulation-invariant queries
that can be answered in $k$-fold exponential time. These turn out to be exactly those that can be
defined in $\phfl{}{}$, the polyadic extension of the higher-order modal fixpoint logic \HFL. This
is genuinely an extension as \polymucalc coincides with $\phfl{0}{}$, the fragment with no higher-order
constructs. Moreover, the level in the exponential-time hierarchy corresponds to the type order
used in formulas: \bisim{$k$-\ExpTime} = $\phfl{k}{}$. 

\begin{figure}
\begin{center}
\begin{tikzpicture}[semithick]
  \path node (ptime)    {\PTime}       -- +(-25:20mm) node (bisptime)    {\bisim{\PTime}}       -- ++(0,15mm) 
        node (exptime)  {\ExpTime}     -- +(-25:20mm) node (bisexptime)  {\bisim{\ExpTime}}     -- ++(0,15mm)
        node (2exptime) {$2$-\ExpTime} -- +(-25:20mm) node (bis2exptime) {\bisim{$2$-\ExpTime}} -- ++(0,15mm)
        node (3exptime) {$3$-\ExpTime} -- +(-25:20mm) node (bis3exptime) {\bisim{$3$-\ExpTime}} -- ++(0,10mm)
        node            {$\vdots$}     -- +(-25:20mm) node               {$\vdots$};
  
  \path[-] (ptime) edge (bisptime) edge (exptime)
           (bisptime) edge (bisexptime)
           (exptime) edge (bisexptime) edge (2exptime)
           (bisexptime) edge (bis2exptime)
           (2exptime) edge (bis2exptime) edge (3exptime)
           (bis2exptime) edge (bis3exptime)
           (bis3exptime) edge (3exptime);
           
  \path ++(8,0)
        node (unknown) {}            -- +(-25:20mm) node (phfl0) {$\phfl{0}{} = \polymucalc$} -- ++(0,15mm) 
        node (solfp)   {$\holfp{2}$} -- +(-25:20mm) node (phfl1) {$\phfl{1}{}$} -- ++(0,15mm)
        node (3holfp)  {$\holfp{3}$} -- +(-25:20mm) node (phfl2) {$\phfl{2}{}$} -- ++(0,15mm)
        node (4holfp)  {$\holfp{4}$} -- +(-25:20mm) node (phfl3) {$\phfl{3}{}$} -- ++(0,10mm)
        node           {$\vdots$}    -- +(-25:20mm) node         {$\vdots$};
           
  \path[-] (phfl0) edge (phfl1) 
           (phfl1) edge (phfl2) edge (solfp)
           (solfp) edge (3holfp)
           (phfl2) edge (phfl3) edge (3holfp)
           (3holfp) edge (4holfp)
           (phfl3) edge (4holfp);

  \path[-] (bisptime)    edge node [pos=.45,above] {\scriptsize \cite{Otto99}} (phfl0)
           (bisexptime)  edge node [pos=.35,above] {\scriptsize \cite{conf/ifipTCS/LangeL14}} (phfl1)
           (bis2exptime) edge node [pos=.35,above] {\scriptsize Thm.~\ref{thm:capture}} (phfl2) 
           (bis3exptime) edge node [pos=.35,above] {\scriptsize Thm.~\ref{thm:capture}} (phfl3)
           (exptime)     edge node [pos=.45,below] {\scriptsize \cite{Imm:lanccc}} (solfp)
           (2exptime)    edge node [pos=.45,below] {\scriptsize \cite{FREIRE201171}} (3holfp) 
           (3exptime)    edge node [pos=.45,below] {\scriptsize \cite{FREIRE201171}} (4holfp)
           ; 
           
  \path +(.9,-2) node {Computational Complexity} +(8.9,-2) node {Descriptive Complexity};
\end{tikzpicture}
\end{center}
\caption{Capturing results for time complexity classes - an overview.}
\label{fig:overview}
\end{figure}
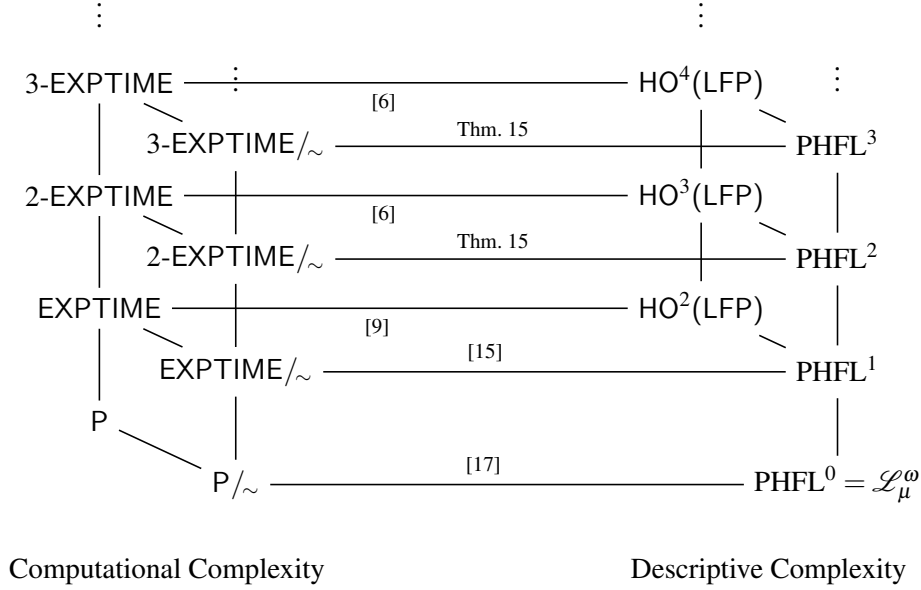

The resulting picture of descriptive time complexity is shown in Fig.~\ref{fig:overview}. A natural
way to extend this is of course to integrate results about the space complexity classes 
\bisim{$k$-\ExpSpace} sitting in between \bisim{$k$-\ExpTime} and \bisim{$(k+1)$-\ExpTime}. Here we
can only state that it is possible to capture these in terms of natural fragments of \PHFL as well
but the proof requires a few more ingredients than the time-complexity case. Here we could use the already
known characterisation of $k$-\ExpTime by Higher-Order Predicate Logics with Least Fixpoints
\cite{FREIRE201171}. For the space complexity classes, we need to first develop such characterisations that
extend the Abiteboul-Vianu Theorem. This can be done \cite{KronenbergerMSc19}, but due to space 
restrictions here its detailed presentation is left for a future publication.

A noteworthy observation for the capturing of \bisim{$k$-\ExpTime} with $k \geq 1$ is that, unlike 
in the case of $k=0$, the requirement of structures being ordered is not necessary for the invocation of the 
(generalised) Immerman-Vardi Theorem. However, starting from $k\geq 1$ the order being provided in the bisimulation-invariant
setting plays a crucial role to emulate existential quantification from the $\HOLFP$ side: the LTS in question 
coming with a $\phfl{0}{2}$-definable order allows us to enumerate sufficiently many sets, functions, etc.\ to
emulate existential quantification. It remains to see in further detail  whether this is an artifact of the proof strategy or a genuine
and inherent part of the correspondence between complexity classes and logics in the bisimulation-invariant framework.

\newpage
\bibliographystyle{eptcs}
\bibliography{literature}

%\appendix

%\input{appendix}

\end{document}